\documentclass[11pt]{article}
\pdfoutput=1
\usepackage{amssymb,epsfig,pspicture}
\usepackage{graphicx}
\usepackage{amsmath}
\usepackage{amsfonts}
\usepackage{latexsym}
\usepackage{enumerate}
\usepackage{amsthm}
\usepackage{graphicx}
\usepackage{caption}
\usepackage{subcaption}
\usepackage{float}
\usepackage{color}
\usepackage{epsfig}
\usepackage{graphicx}
\usepackage{pspicture}

\title{Fair Sets of Some Classes of Graphs\thanks{Pre-print, to be submitted elsewhere.}
\author{Ram~Kumar R\thanks{The work of the author was supported by the University Grants Commission (UGC),
Govt. of India, under their FIP scheme.}\\
        Research~Scholar\\
        Department of Computer~Applications\\
        Cochin University of Science and Technology\\
        e-mail:ram.k.mail@gmail.com \and
        Kannan Balakrishnan\\
         Associate~Professor\\
         Department of Computer~Applications\\
           Cochin University of Science and Technology\\
           e-mail:mullayilkannan@gmail.com \and
           Prasanth G. Narasimha-Shenoi\\Department of Mathematics\\
           Government College, Chittur\\
           Palakkad-678104, India\\
           e-mail:prasanthgns@gmail.com}}
           \date{ }

\newtheorem{Theorem}{Theorem}
\newtheorem{Lemma}{Lemma}
\newtheorem{Corollary}{Corollary}
\newtheorem{Proposition}{Proposition}
\newtheorem{example}{Example}
\newtheorem{conjecture}{Conjecture}
\begin{document}
\maketitle
\begin{abstract}
Given a non empty set $S$ of vertices of a graph, the partiality of a vertex with respect to $S$ is the difference between maximum and minimum of the distances of the vertex to the vertices of $S$.  The vertices with minimum partiality constitute the fair center of the set. Any vertex set which is the fair center of some set of vertices is called a fair set. In this paper we prove that the induced subgraph of any  fair set is connected in the case of trees and characterise block graphs as the class of chordal graphs for which the induced subgraph of all fair sets are connected. The fair sets of $K_{n}$, $K_{m,n}$, $K_{n}-e$, wheel graphs, odd cycles and symmetric even graphs are identified. The fair sets of the Cartesian product graphs are also discussed.
 \end{abstract}

 \section{Introduction}
The main objective in any location theory problem is to identify the ideal locations for setting up a  facility for a set of customers. This problem is approached in different ways. The first one is called the efficiency oriented model where the objective is to minimise the sum of the total distance covered. This corresponds to the median problem in graph theory\cite{key2,key1,rkkb-12}. When we are looking for locating undesirable facilities such as nuclear reactors, garbage dumps etc \cite{key21,key22,key23} the objective becomes maximization of the sum of distances instead of minimization. This corresponds to the antimedian problem in graph theory \cite{key25,key24}.  Another approach is the effectiveness oriented model and this corresponds to the center problem in graph theory. This model is useful in locating emergency facilities such as ambulance, fire stations etc, \cite{key3}. The center problem in graph theory has been studied extensively and various variants of centers have been introduced from time to time \cite{key5,key4}. We also have the anicenter problem in graph theory \cite{chang} which is again used in obnoxious facility location problems. A third approach is the equity oriented model where equitable locations are chosen, that is locations which  are more or less equally fair to all the customers. Issue of equity is relevant in locating public sector facilities where distribution of travel distances among the recipients of the service is also of importance.

Most of the equity based study of location theory concentrate either on comparisons of different measures of equity \cite{key11} or on giving algorithms for finding the equitable locations\cite{key6,key7,key8,key9,key10}.  Also in many optimization problems, we have a set of optimal vertices.  If we want to choose among these, one of the important criteria can be equity or fairness. Taking a leaf out of these literature we define an equity measure called partiality(termed as range in \cite{key10}) and the sets of equitable locations corresponding to this measure in graphs, defined to be fair sets. Rest of the paper is divided as follows, In the section preliminaries, we have some definitions, and results required for the paper.  In section~\ref{cfs}, we prove a result regarding connectedness of fair sets in trees, extend the result to block graphs, a generalization of tress, and finally charcterise chordal graphs, whose fair sets are connected. In Section \ref{fscg}, we discuss fair sets in some graph families namely complete graphs, complete bipartite graphs, odd cycles,  wheel graphs and symmetric even graphs. We also have some partial results in the case of Cartesian Product of graphs in Section~\ref{cp}. In the last section~\ref{con} we sum up giving the possibilities of some future work.

\section{Preliminaries}
We consider only finite simple undirected connected graphs. For the graph $G$, $V(G)$ denotes its vertex set and $E(G)$ denotes its edge set. When the underlying graph is obvious we use $V$ and $E$ for $V(G)$ and $E(G)$ respectively. For two vertices $u$ and $v$ of $G$, distance between $u$ and $v$ denoted by $d(u,v)$, is the number of edges in the shortest $u-v$ path. A vertex $v$ of a graph $G$ is called a cut-vertex if $G-v$ is no longer connected. Any maximal induced subgraph of $G$ which does not contain a cut-vertex is called a \emph{block} of $G$. A graph G is a \emph{block graph} if every block of $G$ is complete. A graph $G$ is \emph{chordal} if every cycle of length greater than three has a chord; namely an edge connecting two non consecutive vertices of the cycle. Trees, $k$-trees, interval graphs, block graphs are all examples of chordal graphs. Chordal graphs form a well studied class of graphs as they have a very unique clique-pasted structure and many of its properties are inherited by the corresponding clique-trees \cite{key13,key14,key15}. For two vertices $u$ and $v$ of $G$, distance between $u$ and $v$ denoted by $d_G(u,v)$ (if $G$ is obvious, then we write $d(u,v)$), is the number of edges in a shortest $u-v$ path. The set of all vertices which are at a distance $i$ from the vertex $u$ is denoted by $N_i(u)$. The \emph{eccentricity} $e(u)$ of a vertex $u$ is $\max\limits_{\substack{v \in V(G)}} d(u,v)$. A vertex $v$ is an \emph{eccentric vertex} of $u$ if $e(u)= d(u,v)$.  A vertex $v$ is an eccentric vertex of $G$ if there exists a vertex $u$ such that $e(u)=d(u,v)$. A graph is a \emph{unique eccentric vertex} graph if every vertex has a unique eccentric vertex. The unique eccentric vertex of the vertex $u$ is denoted by $\bar{u}$. The diameter of the graph $G$, diam(G), is $\max\limits_{\substack{u \in V(G)}} e(u)$. Two vertices $u$ and $v$ are said to be diametrical if $d(u,v)=diam(G)$.  The interval $I(u,v)$ between vertices $u$ and $v$ of $G$ consists of all vertices which lie in some $u-v$ shortest path.  A vertex $u$ of a graph $G$ is called a universal vertex if $u$ is adjacent to all other vertices of $G$. The \emph{Cartesian product} $G\Box H$ of two graphs $G$ and $H$ has vertex set, $V(G) \times V(H)$, two vertices $(u,v)$ and $(x,y)$ being adjacent if either $u=x$ and $vy \in E(H)$ or $ux \in E(G)$ and $v=y$. For more on graph products see \cite{key20}.\\ 
For any $x \in V$ and $S \subseteq  V$, $min(x,S)$ $(max(x,S))$ denote the minimum(maximum) of the distances between $x$ and the vertices of $S$. For an $x \in V$ and $S \subseteq V$, $ max(x,S)-min(x,S)$ is defined to be the \emph{partiality} of $x$ with respect to $S$ and is denoted by $f_G(x,S)$.  For a given nonempty vertex set $S$, $\{v \in V: f_G(v,S) \le f_G(x,S)$ $\forall x \in V\}$ is defined as the \emph{fair center} of $S$ and is denoted by $FC(S)$. Any $A \subseteq V$ such that $A=FC(S)$ for  some $S \subseteq V$, $|S|>1$, is called a \emph{fair set} of $G$. When the underlying graph is obvious we use the notation $f(x,S)$ instead of $f_G(x,S)$.  When $|S|=1$, we can easily see that for all $x\in V$, $f(x,S)=0$ so, $FC(S)=V$.

\section{Graphs with Connected Fair Sets}\label{cfs}
In this section we characterize those chordal graphs for which the subgraph induced by fair sets are connected.
\begin{Lemma}
 For any tree $T$, the subgraph induced by any fair set is connected.
 \label{lem1}
\end{Lemma}
\begin{proof} Let $A$ be a fair set with $A=FC(S)$ where $S=\{v_1,\dots,v_k\}$. Let $u,v \in A$.
Assume that $ v_1,\ldots,v_k $ are such that $d(u,v_1) \le \dots\le d(u,v_k) $. Let $ P$ be the path $u u_1 \dots u_m v$.
At each stage as we move from $u$ to $v$ through the path $P$, let $d_1,\ldots,d_k$ denote the distance between the corresponding vertex of the path and 
$v_1,\dots,v_k$ respectively. At $u$, $f(u,S) =d_k-d_1$. 
Since in any tree, the distances of two adjacent vertices from a given vertex differ by one, we have $f(u_1,S)$ is either $f(u,S)$ or $f(u,S)+ 1$ or $f(u,S)+ 2$.  To prove $f(u_1,S)=f(u,S)$, we consider the following cases.  

\textbf{Case 1}: $f(u_1, S) =f(u, S)+2$.\\
We first consider $f(u_1, S) =f(u, S)+ 2$.  This is possible only when $d_k$ increases by one and $ d_1 $ decreases by one as we move from $u$ to $ u_1$.  Therefore, as we traverse from $u$ to $v$ through $P$, and the graph is a tree, $d_k$ always increase by $1$, so that the partiality cannot decrease.  Hence $f(v, S)>f(u, S)$ which is a contradiction to the assumption that $u,v\in A$. 

\textbf{Case 2}: $f(u_1, S)=f(u,S)+1$.

\textbf{Subcase 2.1:} As we move from $u$  to $ u_1$, $d_k$ increases by one and the role of $v_1$ is taken by some other vertex say $v_2$.  Then similar to the Case 1, we can see that $f(v,S)>f(u,S)$, and a contradiction is obtained.

\textbf{Subcase 2.2:} The role of $v_k$ is taken by another vertex, (say) $v_{k-1}$, so that the maximum distance remains the same(here $d_{k-1}$) and $d_1$
decreases by one. Now as we move from $u_1$ to $ u_2$, since there was an increase in $d(u_1,v_{k-1})$ as compared to $d(u,v_{k-1})$  the maximum
distance keeps on increasing so that the partiality becomes non decreasing. Hence the $f(v, S)>f(u, S)$, a contradiction to our assumption that $u,v\in A$.

From the contradictions of Cases 1 and 2, we obtain $f(u,S)=f(u_1,S)$.  So that $u_1\in A$ and in a similar fashion we can show that $V(P)\subseteq A$. 
Since $u$ and $v$ are arbitrary vertices of $A$, we can see that $A$ is connected.  Hence, we have the lemma.
\end{proof}
\begin{Proposition}
In a block graph the induced subgraph of any fair set is connected.
\label{cor1}
\end{Proposition}
\begin{proof} Let $G=(V,E)$  be a block graph. Let $v_1, v_2 ,\ldots,v_n $ be the vertices of $G$. Let $B_1,B_2,\ldots ,B_r$ be the blocks of $G$. From this form the skeleton $S_G$ of G whose vertices are $ v_1,\ldots,v_n ,B_1 ,\ldots,B_r.$ Here we draw an edge between $v_i$ and $B_j $ if $v_i \in B_j.$ For any block graph $G$, its skeleton $S_G$ is  a tree \cite{key16}.

\begin{figure}[H]
\begin{center}
 \includegraphics[scale=0.15]{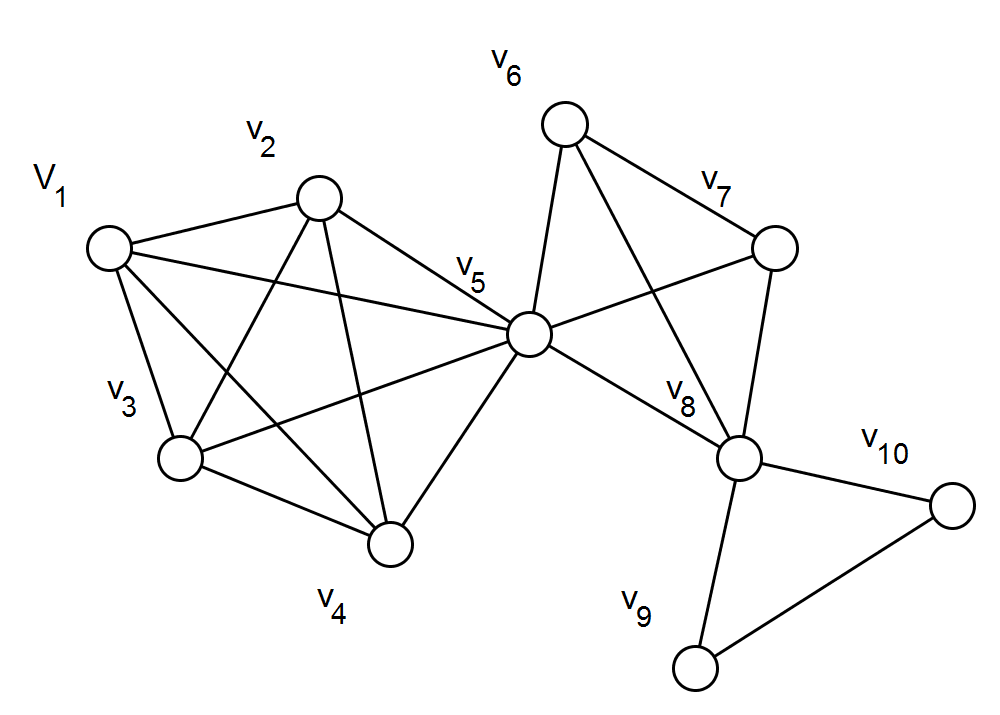}
     \includegraphics[scale=0.15]{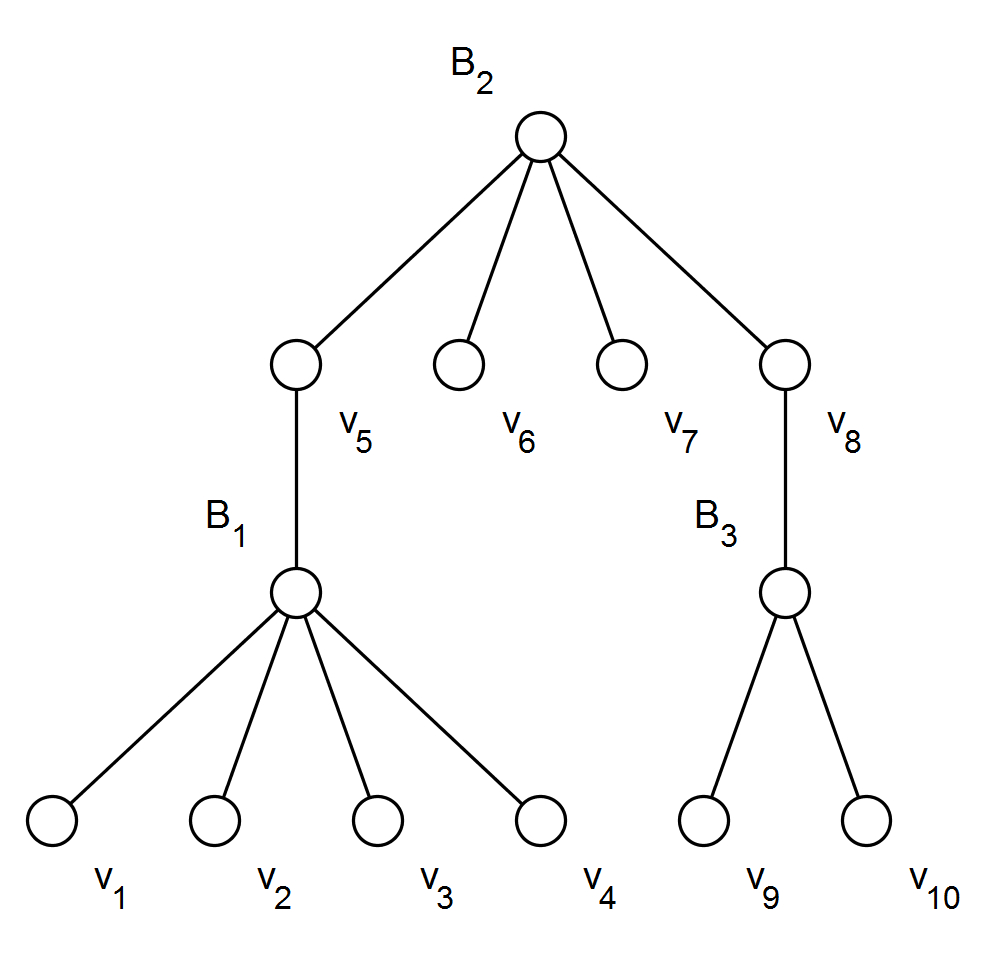}\label{blockkskel}
     \caption{A block graph and its skeleton graph}
     \end{center}
       \end{figure}
Also if $d_G(v_i,v_j)=d$ then $d_{S_G}(v_i,v_j)=2d$. If $S=\{v_1,\ldots,v_k \}$ is a subset of $V(G)$, then for any vertex $v_i$, partiality $f_G(v_i,S)=\frac{1}{2}f_{S_G}(v_i,S)$.  Hence if $v_l \in FC(S)$ with $f_G(v_l,S)=p$, then $f_{S_G}(v_l,S)=2p$. Also for every $v_i \ne v_l$, $f_{S_G}(v_i,S)  \ge 2p$. Now, let  $v_m$ be another vertex in $ G$ such that $f_G(v_m,S)=p$. Then $ f_{S_G}(v_l,S)=2p$, $f_G(v_m,S)=2p$ and $f_{S_G}(v_i,S)\ge 2p$ for every $i=1,\ldots,n$.  Since $S_G$ is connected there exists one path connecting $v_l$ and $v_m$ in $S_G$, say $v_l B_l v_{l+1} B_{l+1}\ldots B_{m-1} v_m$. Since we know that in a tree as we move along a path once partiality increases it cannot decrease $f_{S_G}(v_i,S)\le 2p ,i=l+1,\ldots,m-1$. But since partiality always greater than or equal to $ 2p$, $f_{S_G}(v_i,S)=2p, i=l,l+1,\ldots,m-1,m.$
Therefore $f_G(v_i,S)=p$, $i=l,l+1,\ldots,m-1,m$. Since $v_l$ and $v_{l+1}$ are adjacent to $B_l$ they belong to same block  in G. Therefore $v_l$ and $v_{l+1}$ are adjacent in $G$. Similarly $v_{l+1}$ and $v_{l+2}$ are adjacent in $G$. Hence we get a path $v_l,v_{l+1},\ldots,v_{m-1} v_m$ in $G$ connecting $v_l$ and $v_m$ all of whose partiality is $p$, the minimum. Therefore induced subgraph of any fair set is connected.
\end{proof}
\begin{Theorem}\cite{key17}
A graph $G$ is chordal if and only if it can be constructed recursively by pasting along complete subgraphs, starting from complete graphs.
\end{Theorem}
\begin{Theorem}\label{block}
Let $G$ be a chordal graph. Then $G$ is a block graph if and only if the induced subgraph of any fair set of $G$ is connected.
\end{Theorem}
\begin{proof}
Suppose $G$ is a block graph. Then by Corollary \ref{cor1}, for any $S \subseteq V$ the induced subgraph of $FC(S)$ is connected. Conversely assume that the induced subgraph of all fair sets of $G$ is connected and let  $G$ is not a block graph. Since $G$ is chordal, there exist two  chordal
graphs $G_1$ and $G_2$ such that $G$ can be got by pasting $G_1$ and $G_2$ along a complete subgraph say, $H$, where $|V(H)|>1$. Then there exists two vertices $u$ and $v$ such that $u \in V(G_1)\setminus V(H)$, $v \in V(G_2) \setminus V(H)$ and $u$ and $v$ are adjacent to all vertices of $H$. Consider the vertex set $V(H)$. Since $u$ and $v$ are adjacent to all vertices of $H$, $f(u,V(H))=f(v,V(H))=1-1=0$. For all $x \in V(H)$, $f(x,V(H))=1$. Hence $FC(V(H))$ contains the vertices $u$ and $v$ and any path from $u$ to $v$ pass through the vertices of $H$ which have partiality one. In other words the induced subgraph of the fair center of $V(H)$ is not connected, a contradiction. Therefore induced subgraph of all fair sets of $G$ is connected implies $G$ is a block graph.
\end{proof}
As an illustration of Theorem~\ref{block}, we have the following example.
\begin{example}
For $V(H)=\{v_3,v_4,v_5\}$, we have $A=FC(V(H))=\{v_2,v_6\}$, the induced subgraph of $A$ is not connected.  
\begin{center}
\begin{figure}[H]\label{chordal}
\centering
\includegraphics[scale=0.25]{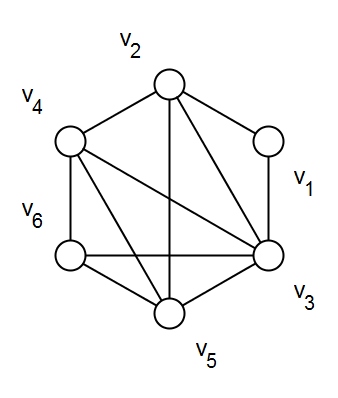}
\caption{$G_1=G[\{v_1,v_2,v_3,v_4,v_5\}]$, $G_2=G[\{v_3,v_4,v_5,v_6\}]$, $H=G[\{v_3,v_4,v_5\}]$}
\end{figure}
\end{center}
\end{example}
\section{Fair sets of some class of graphs}\label{fscg}
In this section, we find out the fair sets of some class of graphs, namely Complete graphs, $K_n-e$, $K_{m,n}$, the wheel graphs $W_n$, odd cycles and, symmetric even graphs.  Before that we have the following lemma.
\begin{Lemma}
For any graph $G=(V,E)$ all the  fair sets $A$ of $G$ are of cardinality either $|V|$ or less than $|V|-1$.
\end{Lemma} 
\begin{proof}
Let $A$ be fair set of $G$ and assume that $A\not= |V|$.  To prove $|A|<|V|-1$.  If possible let $|A|=|V|-1$. Let $A=FC(S)$ where $S \subseteq V$. Let $y$ be the vertex which is not in $A$.  For each $x \in A$ let $f(x,S)=k$.  Also we have $f(y,S)>k$. \\
If $y \in S$ then we have $min(y,S)=0$.  So we must have $max(y,S)>k$.  Therefore there exists an $z$ in $S$ such that $d(y,z) >k$ and this implies that $z \notin A$, a contradiction to the fact that $|A|=|V|-1$. Hence for each $x$ in $S$, $f(x,S)=k$.\\
Next let $y \notin S$. Let $min(y,S)=r$ and $max(y,S)=k+r+s$ where $r,s>0$. Since $min(y,S)=r$ there exists a vertex $w$ adjacent to $y$ such that $min(w,S)=r-1$. Since $f(w,S)=k$ we have $max(w,S)=k+r-1=k+r+s-(s+1)=max(y,S)-(s+1)$. Since $s\ge 1$, we have $|max(y,S)-max(w,S)|\ge 2$, a contradiction. 
\end{proof}

\begin{Proposition}\label{complete}
For the complete graph on $n$ vertices $K_n$, any $A\subseteq V$ such that $|A| \neq n-1$, is a fair set.
\end{Proposition}
\begin{proof}Let $S \subseteq V$ with $|S| >1$. Then for every $x \in S$, $f(x,S)= 1-0 =1$ and for every $y \notin S$, $f(y,S)=1-1=0$. Therefore $FC(S)= S^{c}$. Also if $|S|=1$ then $FC(S)= V$. Hence all $A \subseteq V$ such that $|A| \neq n-1$ is a fair set.
\end{proof}
We get a different result if we delete an edge from a complete graph.  
\begin{Proposition}
Let $G$ be the graph $K_n-e$ with $V(G)=\{v_1,\ldots,v_n\}$ and let $e$ be the edge $v_1v_2$. Then $A \subseteq V$ is a fair set if and only if $|A| \neq n-1$ and either $\{v_1,v_2\} \subseteq A$ or $\{v_1,v_2\} \subseteq A^{c}$.
\end{Proposition}
\begin{proof}Let $\{v_1,v_2\} \subseteq A$ with $|A| <n-1$. Then $|A^{c}| \ge 2$. For each $x \in A$, $ f(x,A^{c})=1-1=0$. For each $x \in A^{c}, f(x,A^{c})=1-0=1$. Therefore $FC(A^{c})=A$. Now, let $\{v_1,v_2\} \subseteq A^{c}$. For each $x \in A$, $f(x,A^{c})=1-1=0$. $f(v_1,A^{c})=2-0=2$, $f(v_2,A^{c})=2-0=2$ and for every other $x$ in $A^{c}$, $f(x,A^{c})=1-0=1$. Hence $FC(A^{c})=A$.\\
Conversely, Let $A$ be a fair set.  We first prove that $|A|\neq n-1$.  If $|A|= n-1$ then $|A^c|=1$  so we have $FC(A^c)=V$.  If $B$ is any set such that $FC(B)=A$ then $|B|>1$.  If $\{v_1,v_2\}\subseteq B$, then $FC(B)=B^c\neq A$.  If $v_1\in B\cap A$ and $v_2\notin B$ then $FC(B)=B^c\setminus \{v_2\}\neq A$.  If $\{v_1,v_2\}\subseteq B^c$, then again $FC(B)=B^c\neq A$.  Hence $|A|\neq n-1$.\\
Now let us assume that there is a set $B$ with $FC(B)=A$.  Suppose 
neither $\{v_1,v_2\} \subseteq A$ nor $\{v_1,v_2\} \subseteq A^{c}$.  With out loss of generality, we assume $v_1\in A$ and $v_2\notin A$. If $\{v_1,v_2\}\subseteq B$, then $FC(B)=B^c\neq A$.   If $v_1\in B\cap A$ and $v_2\notin B$ then $FC(B)=B^c\setminus \{v_2\}\neq A$.  If $\{v_1,v_2\}\subseteq B^c$, then again $FC(B)=B^c\neq A$.  From these we arrive at a contradiction to our assumption that $FC(B)=A$.  Hence our supposition is wrong so either $\{v_1,v_2\} \subseteq A$ or $\{v_1,v_2\} \subseteq A^{c}$.  Hence the proposition.
\end{proof}

Now, we prove case for complete bipartite graph $G=K_{m,n}$ with bipartition $(X,Y)$ where $|X|=m$ and $|Y|=n$.  We assume $A=A_1 \cup A_2$ where $A_1 \subseteq X$ and $A_2 \subseteq Y$.  
\begin{Proposition}\label{completebip}
Let $G$ be a complete bipartite graph $K_{m,n}$ with bipartition $(X,Y)$ where $|X|=m$ and $|Y|=n$. 
Let $A=A_1 \cup A_2$ where $A_1 \subseteq X$ and $A_2 \subseteq Y$. Then $A$ is a fair set if and only if  $|A_1| \neq m-1$ and $|A_2| \neq n-1$.
\end{Proposition}
\begin{proof}
We prove the proposition case by case.\\
\textbf{Case 1:} $|A_1|<m-1$ and $|A_2| <n-1$.\\
Then $A^{c}=(X-A_1) \cup (Y-A_2)$. For, each $x \in A^{c}$, $f(x,A^{c})=2-0=2$ and for each $x \in A$ $f(x,A^{c})=2-1=1$. So, $FC(A^{c})=A$.\\
\textbf{Case 2:} $A=X \cup Y $. \\
We can see that $FC(A)=A$.\\
\textbf{Case 3:} $|A_1|=m$ and $|A_2| < n-1$.\\
Then as in the Case 1, we have $FC(A^{c})=A$.\\
\textbf{Case 4:} $|A_1|=m-1$ and $|A_2|=n-1$.\\
Here $|A^{c}|=2$ let it be $\{x_m,y_n\}$ where $x_m \in X$ and $y_n \in Y$. For each  $x \in A_1$, $f(x,A^{c})=2-1=1$, for each $x \in A_2$, $f(x,A^{c})=2-1=1$. $f(x_m,A^{c})=f(y_n,A^{c})=1-0=1$. So  $FC(A^{c}) = X \cup Y$.\\
\textbf{Case 5:} $|A_1|=m-1$ and $|A_2|<n-1$.\\
Let $A_1=X \setminus \{x_1\}$. For each  $x \in A_1$, $f(x,A^{c})=2-1=1$, for each $x \in A_2$, $f(x,A^{c})=2-1=1$. $f(x_1,A^{c})=1-0=1$ and for each  $x \in Y \setminus A_2$, $f(x,A^{c})=2-0=2$. So $FC(A^{c})= A_1 \cup A_2 \cup \{x_1\}=X \cup A_2$.\\
\textbf{Case 6:} $|A_1|=m-1$ and $|A_2|=n$.\\
Then, we  $FC(A^{c})=X \cup Y$.\\
We can easily see that the cases $1$ to $6$, discusses the fair centers of all types of subsets of $V$.  
Hence the proposition.
\end{proof}
As a simple illustration of Proposition~\ref{completebip}, we have an example which discuss the Case 1 of the proposition. 
\begin{example}
$G=K_{5,4}$, with partitions $X=\{x_1,x_2,x_3,x_4,x_5\}$ and $Y=\{y_1,y_2,y_3,y_4\}$.  By choosing $A_1=\{x_1,x_2,x_3\}$, $A_2=\{y_1,y_2\}$, $A^{c}=\{x_4,x_5,y_3,y_4\}$, we can see that $FC(A^c)=A$.
\begin{figure}[H]
\centering
\includegraphics[scale=0.2]{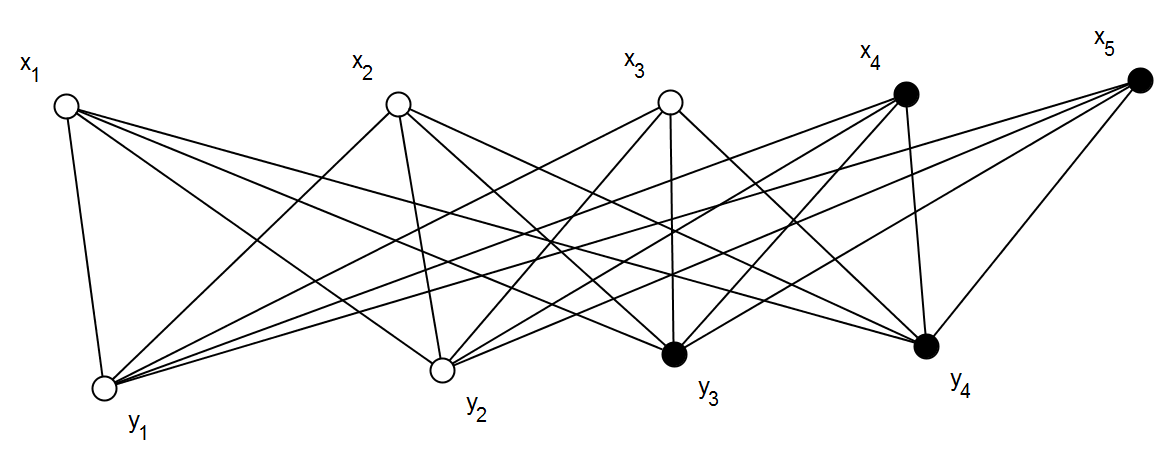}
\label{I}\caption{$K_{5,4}$}
\end{figure}
\end{example}
Now we consider the case when the graph is a wheel $W_n$.  We first prove the case when $n>6$.  
\begin{Theorem}\label{wheel}
 Let $W_{n}$, ($n\ge 6$) be the wheel graph with vertex set $\{v_1,\ldots,v_{n-1},v_n\}$, where $v_n$ is the universal vertex. Let $C_{n-1}$ be the cycle induced by $\{v_1,\ldots,v_{n-1}\}$. Then the fair sets of $W_{n}$ are
    \begin{enumerate}
    \itemsep-12pt
    \item$\{v_i\}$, $ 1 \le i \le n $, \\
\item $\{v_i,v_j\}$ such that $v_i,v_j \in V(C_{n-1})$, $d_{C_{n-1}}(v_i,v_j)=2$,\\
\item $V(W_n)$,\\
\item All sets of the form $A_1\cup \{v_n\}$ where $A_1 \subset V(C_{n-1})$ and $A_1$ is not an induced path of length greater than $n-6$. 
 \end{enumerate}
 \end{Theorem}
 \begin{proof}
 We prove the theorem first for $n>6$ and for the case $n=6$, the proof is similar to the case of $n>6$, so we omit the proof when $n=6$.\\
When $n>6$:\\
We use the notation $v_{i+k}$(or $v_{i-k}$) for $v_{i+k-(n-1)}$(or $v_{i-k+(n-1)}$) when $i+k > n-1$(or $i-k<1)$. First, we prove that the four types of sets described in the theorem are indeed fair sets.
\vspace{-1mm}
\begin{enumerate}[$1.$]
\itemsep-7pt
\item Let $S=\{v_{i-1}, v_n,v_{i+1}\}$, $1\le i\le n-1$. $f(v_i,S)=0$ and for all $u$ other than $v_i$ we have $f(u,S)>0$ so that $FC(S)=\{v_i\}$.  For $S=V$, $f(v_n,S)=1$ and for all $u$ other than $v_n$ we have $f(u,S)=2$, so in this case we can see that $FC(S)=\{v_n\}$. Hence $\{v_i\}$, $1\le i\le n$ are all fair sets. \\
\item Let $S= \{v_n,v_i\}$, $1\le i\le n-1$. $f(v_{i-1},S)=f(v_{i+1},S)=0$ and for all other $u$, $f(u,S) >0$. Hence $FC(S)=\{v_{i-1},v_{i+1}\}$. In other words any $\{v_i,v_j\}$ such that $v_i,v_j \in V(C_{n-1})$, $d_{C_{n-1}}(v_i,v_j)=2$ is a fair set.\\
\item Let $S=\{v_i,v_{i+1},v_n\}$. $f(v_i,S)=f(v_{i+1},S)=f(v_n,S)=1-0=1$. For all other $u$, $f(u,S)=2-1=1$. Hence $FC(S)=V$.\\
\item Now let $S\subseteq V$ be such that $v_n \in S$ and $S$ contains atleast one pair of vertices $v_i$ and $v_j$ such that $d_{C_{n-1}}(v_i,v_j)>2$. Then for every $u \in S$ such that $u \neq v_n$, $f(u,S)=2-0=2$, $f(v_n,S)=1-0=1$ and for every $v \notin S$, $f(v,S)=2-1=1$. Hence $FC(S)=S^c\cup \{v_n\}$. This gives us that for any $A\subseteq V$ such that $v_n\in A$ and $V(C_{n-1})\setminus A$ is none of the following subsets of $V(C_{n-1})$, namely
\begin{enumerate}
\itemsep-10pt
\item $\{v_i\}$, $1\le i\le n-1$.\\
\item $\{v_i,v_{i+1}\}$, $1\le i\le n-1$.\\
\item $\{v_i,v_{i+2}\}$, $1\le i\le n-1$.\label{itemc}\\
\item $\{v_i,v_{i+1},v_{i+2}\}$, $1\le i\le n-1$
\end{enumerate}is a fair set.

The set $A_1 \cup \{v_n\}$ where $A_1$ is the compliment in $V(C_{n-1})$ of a set mentioned in \ref{itemc} above, is the fair center of $\{v_i,v_{i+1},v_{i+2},v_n\}$.
\vspace{-1mm}
\end{enumerate}
Therefore the only sets containing $\{v_n\}$ which have not been identified as fair sets are sets of the type $A_1 \cup \{v_n\}$ where $A_1$ is a path of length greater than $n-6$.
Now let $A=A_1 \cup \{v_n\}$ where $A_1 \subseteq V(C_{n-1})$ forms a path of length greater than $n-6$. Assume there exists an $S \subseteq V$ such that $FC(S)=A$. If $S \subseteq V(C_{n-1})$ then $f(v_n,S)=0$ and it is impossible to have $f(u,S)=0$ for every $u \in A_1$. Hence $S$ cannot be a subset of $V(C_{n-1})$ or $v_n \in S$. If $S$ contains a $v_i$ and $v_j$ of $C_{n-1}$ so that $d_{C_{n-1}}(v_i,v_j)>2$ then $FC(S)= V(C_{n-1}) \setminus S \cup \{v_n\}$. So $FC(S)= A$ implies $V(C_{n-1}) \setminus S= A_1$. We have that $v_i$ and $v_j$, two vertices such that $d_{C_{n-1}}(v_i,v_j)>2$, does not belong to $V(C_{n-1}) \setminus S$. By the choice of $A_1$ such two vertices cannot be simultaneously absent from $A_1$. Hence $V(C_{n-1}) \setminus S \neq A_1$, a contradiction. So assume $S$ does not contain two vertices $v_i$ and $v_j$ such that $d_{C_{n-1}}(v_i,v_j)>2$. Hence $S$ should be any one of the following:
\begin{enumerate}[a)]
\itemsep-7pt
\item $\{v_i,v_n\}$, $1\le i\le n-1$.\\
\item $\{v_i,v_{i+1},v_n\}$, $1\le i\le n-1$.\\
\item $\{v_i,v_{i+2},v_n\}$, $1\le i\le n-1$.\\
\item $\{v_i,v_{i+1},v_{i+2},v_n\}$, $1\le i\le n-1$
\end{enumerate}
But we have already found out the fair centers of all these sets and  none of them have $A$ as its fair center. Hence an $A=A_1 \cup \{v_n\}$ where $A_1 \subseteq V(C_{n-1})$ forms a path of length greater than $n-6$ is not fair set.
Now we have the following observations
 \begin{enumerate}[$A.$]
 \itemsep-10pt
 \item For an $S \subset V(C_{n-1})$, $v_n \in FC(S)$\\
  \item If  $S$ contains $v_i$ and $v_j$ of $C_{n-1}$ where $d_{C_{n-1}}(v_i,v_j)>2$ then $v_n \in FC(S)$\\
 \item Any $A=A_1 \cup \{v_n\}$ is fair set if and only if $A_1$ is not a path of length greater than $n-6$ in $V(C_{n-1})$.\label{itemC}\\
 \item For any $S \subseteq V$ such that $v_n \in S$ and $v_i,v_j \in S$ $\implies$ $d_{C_{n-1}}(v_i,v_j)\leq 2$, fair centers are sets of any of the following forms.\label{itemD}
 \vspace{-1mm} 
 \begin{enumerate}[$i)$]
 \itemsep-10pt
 \item $\{v_i\}$ \\
 \item $\{v_i,v_{i+2}\}$\\
 \item $V(W_n)$\\
 \item $\{v_1,\ldots v_i,v_{i+2},\ldots, v_{n-1},v_n\}$, a set of type described in \ref{itemC}.
 \end{enumerate}
 \end{enumerate}
  Hence the only fair sets of $W_n$ are those described in \ref{itemC} and \ref{itemD} above. Hence the theorem.
 \end{proof}
When $n=4$, we can see that $W_4=K_4$, so we prove the case when the graph is a wheel $W_n$ for $n=5$, where we get a proposition, which is entirely different from the Theorem~\ref{wheel}.
\begin{Proposition}
If $G$ is $W_5$ with $V=\{v_1,v_2,v_3,v_4,v_5\}$, where $v_5$ is adjacent to all other vertices and $v_1v_2v_3v_4v_1$ is the outer cycle, then the only fair sets of $G$ are $\{v_5\},\{v_1,v_3\}, \{v_2,v_4\}, \{v_1,v_3, v_5\}, \{v_2,v_4,v_5\}$ and $V$. 
\end{Proposition}   
\begin{proof}
Given a non empty vertex set $S$, let $d_1,d_2,\ldots d_k$ be the distances of the vertex $v_1$ from the vertices of $S$ where $d_1 \le d_2 \le \ldots \le d_k$. Then the distances of $v_3$ from vertices of $S$ are $2-d_k,2-d_{k-1},\ldots, 2-d_1$ where $2-d_k \le 2-d_{k-1}\le ,\ldots, \le 2-d_1$. Hence $f(v_3,S)=d_k-d_1=2-d_1-(2-d_k)=f(v_1,S)$. Hence if $A$ is any fair set, $v_1 \in A$ implies $v_3 \in A$. Similarly $v_2 \in A$ implies $v_4 \in A$. Now $f(v_i,V)=2$ for $1 \le i \le 4$ and $f(v_5,V)=1$. Hence $FC(V)=\{v_5\}$. Similarly we can observe that $FC(\{v_5,v_4\})=\{v_1,v_3\}$, $FC(\{v_5,v_3\})=\{v_2,v_4\}$, $FC(\{v_1,v_3\})=\{v_2,v_4,v_5\}$, $FC(\{v_2,v_4\})=\{v_1,v_3,v_5\}$ and $FC(\{v_1,v_2,v_5\})= V$. Hence the fair sets of $W_5$ are precisely those described in the theorem.
\end{proof} 
\begin{Theorem}\label{oddcycle}
Let the graph $G=C_{2m+1}$ be an odd cycle with vertex set $V=\{v_1, \ldots, v_{2m+1}\}$. $A \subseteq V$ is a fair set if and only if a pair of consecutive vertices say, $v_1$ and $v_{2}$ is in $A$ implies $v_{m+2}$, the vertex which is eccentric to both $v_1$ and $v_{2}$, is also in $A$.
\end{Theorem}
\begin{proof}
Let $A=FC(S)$ with $v_1,v_2 \in A$. Let $min(v_1,S)=d_{min}$ and $max(v_1,S)=d_{max}$. Since $v_1, v_2 \in A$, $f(v_1,S)= f(v_2,S)$. We shall consider here two different cases.\\
\textbf{Case 1}: $min(v_2,S)=d_{min}+1$ and $max(v_2,S)=d_{max}+1$. Then there exists a vertex $v \in S$ such that $d(v_1,v)=d_{max}$ and $d(v_2,v)=d_{max}+1$. Then $d(v_{m+2},v)= m-d_{max}$. If there exists a vertex $v'\in S$ such that $d(v_{m+2},v')<m-d_{max}$ then $d(v',v_1)>d_{max}$, a contradiction. Hence $min(v_{m+2},S)=d(v_{m+2},v)$. Similarly there exists a vertex $u$ such that $d(u,v_1)=d_{min}$ and $d(u,v_2)=d_{min}+1$. $max(v_{m+2},S)=d(v_{m+2},u)=m-d_{min}$. Therefore $f(v_{m+2},S)=m-d_{min}-(m-d_{max})=d_{max}-d_{min}=f(v_1,S)=f(v_2,S)$. \vspace{2mm}That is $v_{m+2} \in A$.\\
\textbf{Case 2}: $min(v_2,S)=d_{min}$ and $max(v_2,S)=d_{max}$. Let $u$ and $u'$ be such that $min(v_1,S)=d(v_1,u)$ and $min(v_2,S)=d(v_2,u')$. $max(v_{m+2},S)= m-d(v_1,u)=m-min(v_1,S)=m-d_{min}$. Let $v$ and $v'$ be such that $max(v_1,S)=d(v_1,v)$ and $max(v_2,S)=d(v_2,v')$. If $v=v'$ then $v=v_{m+2}$. In this case $min(v_{m+2},S)=0$. Therefore $f(v_{m+2},S)=m-d_{min}=max(v_1,S)-min(v_1,S)=f(v_1,S)$. If $v \neq v'$, $d(v_2,v)=d(v_1,v')=d_{max}-1$. Hence $min(v_{m+2},S)=d(v_{m+2},v)=d(v_{m+2},v')=m-(d_{max}-1)=m-d_{max}+1$. Therefore $f(v_{m+2},S)=m-d_{min}-(m-d_{max}+1)=d_{max}-d_{min}-1 < f(v_1,S)=f(v_2,S)$. This contradicts the fact that $v_1,v_2 \in FC(S)$ and so we rule out this possibility. Hence in all the possible cases $v_1,v_2 \in A \Rightarrow v_{m+2} \in A$.\vspace{1.5mm}\\
Conversely, assume that $A \subseteq V$ is such that for every pair of consecutive vertices $v_i ,v_{i+1}$ in $A$, $v_{m+i+1}$ belong to $A$. Let $v_1,\dots v_k$, $k>1$,  be consecutive vertices belonging to $A$. Then $v_{m+2},v_{m+3},\dots,v_{m+k}$ belong to $A$. With out loss of generality we may assume that
\vspace{-2mm}
\begin{enumerate}[$i$)]
\itemsep -14pt
\item  $A$ does not contain any consecutive set of vertices other than the above two.\\
\item $v_{m+1}$ does not belong to $A$.
\end{enumerate}
\vspace{-1mm}
Now, construct the set $S$ as follows
\vspace{-1mm}
\begin{enumerate}[$step$ $I)$]
\itemsep -12pt
\item If $k=3r$ or $3r+1$ for some integer $r$ then add vertices $v_2,v_5,\ldots,v_{3r-1}$ of $A$ to $S$. If $k=3r+2$ then add to $S$ the vertices $v_3,v_6, \ldots,v_{3r}$ from $A$.\\
\item  Add to $S$ the vertices $v_i, m+2 \le i \le m+k$ of $A$, which are not an eccentric vertex of any of the vertices in $S$.\\
\item  Add $A^c$ to $S$.
\end{enumerate}
Let $x \in V \setminus A$.  Then $x \in S$ and therefore $min(x,S)=0$. Let $y$ and $z$ be the eccentric vertices of $x$ in $G$.  Take note that $yz$ is an edge.  If $\{y,z\}\subseteq S^c$, then we have $\{y,z\}\subseteq A$.  Since $x$ is an eccentric vertex of $y$ and $z$, we have $x\in A$, which is not true.  Hence either $y$ or $z$ belongs to $S$, and we have $max(x,S)=m$, so that $f(x,S)=m$.

Let $x \in A$ be such that both the neighbours of $x$ do not belong to $A$. Then $x \notin S$ and neighbours of $x$ belong to $S$ and $min(x,S)=1$. Let $x_1$ and $x_2$ be the eccentric vertices of $x$. Then $x_1,x_2 \notin S$ implies either $x_1,x_2 \in \{v_1,\ldots,v_k\}$ or $x_1,x_2 \in \{v_{m+2},\ldots,v_{m+k}\}$. In the former case $x \in \{v_{m+1}, \ldots, v_{m+k}\}$ and in the latter case $x \in \{ v_1,\ldots,v_k\}$ and this is not possible by the choice of $x$. Hence either $x_1$ or $x_2$ belong to $S$. Hence $max(x,S)=m$. Therefore $f(x,S)=m-1$. By the way of choice of vertices $v_i,1 \le i \le k,$ in $S$ either $min(v_i,S)=1$ and $max(v_i,S)=m$ or $min(v_i,S)=0$ and $max(v_i,S)=m-1$. Hence $f(v_i,S)=m-1$ for $1 \le i \le k$. For $m+2 \le i \le m+k$, $v_i \notin S$ implies eccentric of $v_i$ belong to $S$. Therefore in this case $min(v_i,S)=1$ and $max(v_i,S)=m$ or  $f(v_i,S)=m-1$. Now, for $m+2 \le i \le m+k$, $v_i \in S$ implies $min(v_i,S)=0$. Now an eccentric vertex of $v_i$, $m+2 \le i \le m+k$, belong to $S$ implies $v_i \notin S$. Hence for $v_i \in S$ eccentric vertices of $v_i \notin S$. Also there are no three consecutive vertices among $v_i$'s, $1 \le i \le k$, absent from $S$. Hence $max(v_i,S)=m-1$ for $m+2 \le i \le m+k$. Also the two eccentric vertices of $v_{m+k+1}$, $v_k$ and $v_{k+1}$ does not belong to $S$. Hence if $v_{m+k+1} \in S$, $f(v_{m+k+1},S)=m-1$. Therefore for each $v_i \in A$, $f(v_i,S)=m-1$ and for each $v_i \notin A$, $ f(v_i,S)=m$. Hence $FC(S)=A$ or $A$ is a fair set.
\end{proof}
We have an immediate corollary for the theorem~\ref{oddcycle} and the proof follows from the proof of the theorem~\ref{oddcycle}.
\begin{Corollary} 
If $G$ is an odd cycle, and $U\subset V(G)$ contains no adjacent vertices then $U$ will be a fair set of $G$.
\end{Corollary}

\begin{Corollary}
The only connected fair sets of an odd cycle $C_{2m+1}$  are singleton (vertex) sets and the whole vertex set $V$.
\end{Corollary}
\begin{proof}
By the Theorem~\ref{oddcycle}, $\{v_i\}, 1 \le i \le 2m+1$ and $V$ are fair sets. Now let $A \subseteq V$ be a connected fair set of $C_{2m+1}$ which contains more than one element. Let $v_i, v_j \in A$. Then there exists a path connecting $v_i$ and $v_j$ in $A$ with out loss of generality we may assume that it is $v_i,v_{i+1},\ldots v_j$. $v_i,v_{i+1} \in A$ implies $v_{m+i+1} \in A$. Therefore a path connecting $v_i$ and $v_{m+i+1}$ lies in $A$. Since this path contains $m+2$ consecutive vertices by the theorem we can conclude that $A$ should also contain the other $m-1$ vertices or $A=V$.
\end{proof}
 \subsection{Fair sets of Symmetric Even graphs}
A graph G is called $even$ if for every $u \in V(G)$ there exists a $v \in V(G)$ such that $d(u,v)=diam(G)$. An even graph $G$ is $symmetric$ if for every $u \in  V(G)$ there exists a $v \in V(G)$ such that $I(u,v)=G$, see Gobel et.al \cite{key18}. Hypercubes, even cycles etc are well known examples of symmetric even graphs.  From the definition of symmetric even graphs it is clear that they are unique eccentric vertex graphs and every vertex of such a graph is an eccentric vertex. 
For more examples of symmetric even graphs and one way of constructing such graphs,  see, Gobel et.al \cite{key18}.  We first write the following proposition.
\begin{Proposition}[Proposition 4 of \cite{key18}]
Let $u$ and $v$ be vertices of a symmetric even graph $G$ of diameter $d$. If $v \in N_{i}(u)$ and $\bar{v} \in N_{j}(u)$, then $i+j=d$.
\label{prop1}
\end{Proposition}
Next we identify the fair sets of symmetric even graphs.
\begin{Theorem}
Let G be a symmetric even graph. An $A \subseteq V$ is a fair set if and only if a vertex $x \in A \Rightarrow \bar{x} \in A$.
\end{Theorem}
\begin{proof} Let $diam(G)=d$. Assume $A \subseteq V$ is a fair set with $FC(S)=A$ where $S=\{v_1, \ldots ,v_k\}$ and and let $x \in A$. Let $d(x,v_i)=d_i, 1 \leq i \leq k$. With out loss of generality we may assume that $d_1 \leq d_2 \leq \ldots \leq d_k$. Then $f(x,S)=d_k-d_1$. By proposition \ref{prop1}, $d(x,v_i)= d_i \Rightarrow d(\bar{x},v_i)=d-d_i$. Therefore $f(\bar{x},S)=d-d_1-(d-d_k)= d_k-d_1$.  Hence $x \in A \Rightarrow \bar{x} \in A$.\\
Conversely, assume that $A \subseteq V$ is such that $x \in A \Rightarrow \bar{x} \in A$. To prove $A=FC(S)$ for some $S \subseteq V$. Let $A =\{x_1, \ldots, x_m,\bar{x_1},\ldots, \bar{x_m}\}$. Let $S_1=S \setminus \{x_1,\ldots,x_m\}$. Suppose for every $y \in v$ either $y \in S_1$ or some neighbour of $y$ is in $S_1$. Then for each $x_i, i \leq i \leq m, f(x_i,S)=d(x_i,\bar{x_i})-1 =d-1$(Here the minimum distance is $1$ since $x_i \notin S$ and some neighbour of $x_i$ is in $ S$). For each $\bar{x_i} ,1 \leq i \leq m$, $f(\bar{x_i},S)=d-1-0= d-1$(Here the maximum distance is $d-1$ since $x_i \notin S$ and some neighbour of $x_i$ is in $S$. The minimum distance is 0 since $\bar{x_i} \in S$). Now, for a $y$ different from $x_i,\bar{x_i}$, $1 \leq i \leq m$, $min(y,S)=0$ since $y \in S$ and $max(y,S)=d$ since $\bar{y} \in S$. Therefore $f(y,S)=d-0=d$. In other words $FC(S)=A$. Now, assume that there exists  vertices  in $V$ such that neither those vertices nor any of their adjacent vertices are in $S$. Assume that $x_1$ is one such vertex. Then $min(x_1,S) >1$ and $max(x_1,S)=d$. Therefore $f(x_1,S)=max(x_1,S)-min(x_1,S)\leq d-2$. Let $x_k$ be vertex such that a neighbour of $x_k$ is in $S$. Then $min(x_k,S)=1$ and $max(x_k,S)=d$ and therefore $f(x_k,S)=d-1$. Therefore $FC(S) \neq A$. Let $S_1= \{x_1\} \cup S \setminus \{\bar{x_1}\}$. Now, $min(x_1,S_1)=0$ and $max(x_1,S_1)= d-1$. If for every $v \in V$ either $v \in S_1$ or some neighbour of $v$ is in $S_1$ then $FC(S_1)=A$. Otherwise continue the above process until we get an $S_i$ such that $FC(S_i)=A$ and that proves the theorem.
\end{proof}
\begin {Corollary}
The only connected fair set of an even cycle $C_{2m}$ is the whole vertex set $V$.
\end{Corollary}
\begin{proof}
Let A be a connected fair set of $C_{2m}$. Let $u \in A$. Then by the above theorem $\bar{u} \in A$. Since $A$ is connected atleast one of the paths connecting $u$ and $\bar{u}$ should be in $A$. Again by the theorem the eccentric vertices of the vertices of this path should also be in $A$. Hence $A=V$.
\end{proof}
\section{Fair sets and cartesian product of graphs}\label{cp}
Next we have an expression for the fair center of product sets in the cartesian product of two graphs
\begin{Theorem}
Let $G_1=(V_1,E_1)$ and $G_2=(V_2,E_2)$ be two graphs. Let $S_1\subseteq V_1$ and $S_2\subseteq V_2$. Then $FC(S_1 \times S_2)= FC(S_1) \times FC(S_2)$ where
$FC(S_1 \times S_2)$ is the fair center of $S_1 \times S_2$ in the graph $G_1 \Box G_2$, $FC(S_1)$ is the fair center of $S_1$ in the graph $G_1$ and
$FC(S_2)$ is the fair center of $ S_2$ in the graph $ G_2$.
\end{Theorem}
\begin{proof}
Let $(x,y)\in V_1 \times V_2$, $S_1=\{u_{11},u_{12},\dots,u_{1l}\}$ and $S_2=\{u_{21},\dots,u_{2m}\}$ where $u_{11}$ is the vertex nearest to $x$ and
$u_{1l} $ is the vertex farthest from $x$ and $u_{21}$ is the vertex nearest to $y$ and $u_{2m}$ is the vertex farthest from $y$. For $(u_{1i},u_{2j}) \in
S_1\times S_2$
\begin{eqnarray*}
 d((x,y),(u_{11},u_{21}))  & = &  d(x,u_{11})+d(y,u_{21})\\
                                     & \le & d(x,u_{1i})+d(y,u_{2j})(=d((x,y),(u_{1i},u_{2j})))\\
                                     & \le & d(x,u_{1l})+d(y,u_{2m})\\
                                     & = & d((x,y),(u_{1l},u_{2m}))
\end{eqnarray*}
That is, if $u_{11}$ is the vertex nearest to $x$ in $S_1 \subseteq V_1$ , $u_{21}$ is the vertex nearest to $y$ in $ S_2 \subseteq V_2,$ $u_{1l}$ is the vertex
farthest from $x$ in $S_1 \subseteq V_1$ and $u_{2m}$ is the vertex farthest from $y$ in $S_2 \subseteq V_2$ then $(u_{11},u_{21})$ is the vertex nearest to
$(x,y)$ in $S_1 \times S_2$ and $(u_{1,l},u_{2,m})$ is the vertex farthest from $(x,y)$ in $S_1 \times S_2$
\begin{eqnarray*}
f_{G_1 \Box G_2}((x,y),S_1 \times S_2) &=& d((x,y),(u_{1l},u_{2m}))-d((x,y),(u_{11},u_{21}))\\
                                                   & = & d(x,u_{1l})+d(y,u_{2m})-d(x,u_{11})-d(y,u_{21})\\
                                                      & = &d(x,u_{1l})-d(x,u_{11}+d(y,u_{2m}-d(y,u_{21})\\
                                                              & = &f_{G_1}(x,S_1)+f_{G_2}(x,S_2)
\end{eqnarray*}
Now, let $u_1 \in FC(S_1)$ where $S_1 \subseteq V_1$ and let $u_2 \in FC(S_2)$ where $S_2 \subseteq V_2$.  That is
$f_{G_1}(u_1,S_1) \le f_{G_1}(x,S_1)$, $\forall x \in V_1$ and $f_{G_2}(u_2,S_2) \le f_{G_2}(y,S_2)$, $\forall y \in V_2$. Therefore $ f_{G_1}(u_1,S_1)+f_{G_2}(u_2,S_2) \le f_{G_2}(x,S_1)+f_{G_2}(y,S_2)$.  So, $f_{G_1 \Box G_2}((u_1,u_2),S_1 \times S_2) \le f_{G_1 \Box G_2}((x,y),S_1 \times S_2)$, $ \forall (x,y) \in V_1 \times V_2$.  
Hence $(u_1,u_2) \in FC(S_1 \times S_2)$  in $G_1 \Box G_2$.\\
Conversely, assume that $(u_1,u_2)\in FC(S_1 \times S_2)$ in $G_1 \Box G_2$ where $S_1 \subseteq V_1$ and $S_2 \subseteq V_2$.  That is, $f_{G_1 \Box G_2}((u_1,u_2),S_1 \times S_2) \le f_{G_1 \Box G_2}((x,y),S_1 \times S_2)$, where $S_1 \subseteq V_1$ and $S_2 \subseteq V_2$, $\forall x \in V_1$, $y \in V_2$.  Therefore, $ f_{G_1}(u_1,S_1)+f_{G_2}(u_2,S_2) \le f_{G_1}(x,S_1)+f_{G_2}(y,S_2)$, $ \forall x \in V_1$ and $y \in V_2 $ or $  f_{G_1}(u_1,S_1)+f{G_2}(u_2,S_2) \le f_{G_1}(x,S_1)+f_{G_2}(u_2,S_2)$, $ \forall x \in V_1$.  That is, $ f_{G_1}(u_1,S_1) \le f_{G_1}(x,S_1)$, $ \forall x \in V_1$.  Hence, $u_1\in FC(S_1)$ in $G_1$. Similarly $u_2 \in FC(S_2)$ in $G_2$. Thus, $FC(S_1 \times S_2)= FC(S_1) \times FC(S_2)$.
\end{proof}
\begin{Corollary}
Let $G_1=(V_1,E_1)$ and $G_2=(V_2,E_2)$ be two graphs. Then the  subgraph induced by  $FC(S_1 \times S_2)$, where $S_1 \subseteq V_1$ and $S_2 \subseteq V_2$
is connected in $G_1 \Box G_2$ if and only if the subgraph induced by $FC(S_1)$ is connected in $G_1$ and the subgraph induced by $FC(S_2)$is connected
in $G_2$.
\end{Corollary}
\section{Conclusion}\label{con}
In this paper we could identify the fair sets of various classes of graphs such as complete, complete bipartite, odd cycle, wheel and symmetric even graphs. Identifying the fair sets of more classes of graphs shall be more interesting and challenging. We also proved some partial results in the case of Cartesian products. It will be interesting to find an expression for fair center of any set in the product graphs in terms of the fair centers in the factor graphs. It is also natural to study the behavior of fair centers in relation to other graph products. We identified block graphs as the only chordal graphs with connected fair sets. We could not find graphs other than block graphs, all of whose fair sets are connected. We conclude by posing the following conjecture.
\begin{conjecture}
A graph G is a block graph if and only if the induced subgraph of every fair set of G is connected.
\end{conjecture}

\end{document}